\documentclass[letterpaper, conference]{ieeeconf}
\IEEEoverridecommandlockouts
\usepackage[english]{babel}
\usepackage[usenames,dvipsnames,svgnames,table]{xcolor}
\usepackage[utf8x]{inputenc}
\usepackage{amsmath}
\usepackage{amsfonts}
\usepackage{graphicx}
\usepackage{float}
\usepackage[section]{placeins}
\usepackage{bm}

\usepackage{wrapfig}
\usepackage{amssymb}

\usepackage{dsfont}
\usepackage{pgf}
\usepackage{tikz}
\usetikzlibrary{arrows,automata}

\usepackage{enumerate}

\tikzstyle{SimpleNetwork}=[draw,circle,minimum width=6pt]

\def \etal {\emph{et al.}}

\newtheorem{theorem}{Theorem}

\newtheorem{proposition}{Proposition}

\newtheorem{example}{Example}
\newtheorem{remark}{Remark}

\usepackage[font=footnotesize]{caption}

\def \bs {\boldsymbol}
\def \mc {\mathcal}

\def\expt{\mathbb{E}}
\def\real{\mathbb{R}}

\def\naturals{\mathbb{N}}
\newcommand{\until}[1]{\{1,\dots, #1\}}

\newcommand{\setdef}[2]{\{#1 \; | \; #2\}}
\newcommand{\seqdef}[2]{\{#1\}_{#2}}

\newcommand{\slfrac}[2]{\left.#1\middle/#2\right.}
\newcommand{\lnn}[1]{%
	\ln  \left(#1\right)%
}
\newcommand{\expp}[1]{%
	\exp \! \left(#1\right)%
}

\newcommand\oprocendsymbol{\hbox{$\square$}}
\newcommand\oprocend{\relax\ifmmode\else\unskip\hfill\fi\oprocendsymbol}

\renewcommand{\theenumi}{(\roman{enumi}}

\newcommand\bit[1]{\textit{\textbf{#1}}}

\title{On Distributed Cooperative Decision-Making in Multiarmed Bandits
\thanks{{This revision provides a correction to the original paper, which appeared in the Proceedings of the 2016 European Control Conference (ECC).  The second statement of Proposition 1, Theorem 1 and their proofs are new.  The new Theorem 1 is used to prove the regret bounds in Theorem 2.} } 
\thanks{This research has been supported in part by ONR grant  N00014-14-1-0635, ARO grant W911NF-14-1-0431, and by the Department of Defense (DoD) through the National Defense Science \& Engineering Graduate Fellowship (NDSEG) Program.}}

\author{Peter Landgren, Vaibhav Srivastava, and Naomi Ehrich Leonard
\thanks{P. Landgren, V. Srivastava and N. E. Leonard are with the Department of Mechanical and Aerospace Engineering, Princeton University, Princeton, NJ, USA,\tt{ \{landgren, vaibhavs, naomi\}@princeton.edu}.}}

\date{\today}

\begin{document}
\maketitle

\begin{abstract}
We study the explore-exploit tradeoff in distributed cooperative decision-making using the context of the multiarmed bandit (MAB) problem. For the distributed cooperative MAB problem, we design the cooperative UCB algorithm that comprises two interleaved distributed processes: (i) running consensus algorithms for estimation of rewards, and (ii) upper-confidence-bound-based heuristics for selection of arms.  We rigorously analyze the performance of the cooperative UCB algorithm and characterize the influence of communication graph structure on the decision-making performance of the group. 
%with communication through consensus over a known graph topology.  Agents can occupy that same arm at the same time with no decrease in reward, i.e. there are no collisions, and each agent attempts to minimize their own cumulative regret.  
%We first characterize properties of the running consensus estimation algorithm used for sharing information and from this develop an individualized bandit strategy that is dependent on the graph structure and results in logarithmic regret for all agents.  Using this strategy we analyze the performance of agents with respect to their position in the network.  
\end{abstract}

\section{Introduction}

Cooperative decision-making under uncertainty is ubiquitous in natural systems as well as in engineering networks.  Typically in a distributed cooperative decision-making scenario, there is assimilation of information across a network followed by  decision-making based on the collective information.  The result is a kind of \emph{collective intelligence}, which is of fundamental interest both in terms of understanding natural systems and designing efficient engineered systems. 

A fundamental feature of decision-making under uncertainty is the \emph{explore-exploit} tradeoff. The explore-exploit tradeoff refers to the  tension between learning and optimizing:  the decision-making agent needs to learn the unknown system parameters (exploration), while maximizing its decision-making objective, which depends on the unknown parameters (exploitation).  

MAB problems are canonical formulations of the explore-exploit tradeoff. In a stochastic MAB problem a set of options (arms) are given. A stochastic reward with an unknown mean is associated with each option. A player can pick only one option at a time, and the player's objective is to maximize the cumulative expected reward over a sequence of choices. In an MAB problem, the player needs to balance the tradeoff between learning the mean reward at each arm (exploration), and picking the arm with maximum mean reward (exploitation). 

MAB problems are pervasive across a variety of scientific communities and have found application in diverse areas including controls and robotics~\cite{VS-PR-NEL:14},  ecology~\cite{JRK-AK-PT:78, VS-PR-NEL:13}, psychology~\cite{reverdy2014modeling}, and communications~\cite{lai2008medium, anandkumar2011distributed}. 
Despite the prevalence of the MAB problem, the research on MAB problems has primarily focused on policies for a single agent. The increasing importance of networked systems warrants the development of distributed algorithms for multiple communicating agents faced with MAB problems. In this paper, we  extend a popular single-agent algorithm for the stochastic MAB problem to the distributed multiple agent setting and analyze decision-making performance as a function of the network structure.

The MAB problem has been extensively studied (see~\cite{SB-NCB:12} for a survey). In their seminal work, Lai and Robbins~\cite{lai1985asymptotically} established a logarithmic lower bound on  the expected number of times a sub-optimal arm needs to be selected by an optimal policy. In another seminal work, Auer~\etal~\cite{PA-NCB-PF:02} developed the upper confidence bound (UCB) algorithm for the stochastic MAB problem that achieves the lower bound in~\cite{lai1985asymptotically} uniformly in time. 
Anantharam~\etal~\cite{VA-PV-JW:87} extended the results of~\cite{lai1985asymptotically} to the setting of multiple centralized players. 

Recently, researchers~\cite{kalathil2014decentralized, anandkumar2011distributed} have studied the MAB problem with multiple players in the decentralized setting. Primarily motivated by communication networks, these researchers assume no communication among agents and design efficient decentralized policies. 
Kar~\etal~\cite{kar2011bandit} investigated the multiagent MAB problem in a leader-follower setting. They designed efficient policies for  systems in which there is one major player that can access the rewards and the remaining minor players  can only observe the sampling patterns of the major player.   The MAB problem has also been used to perform empirical study of collaborative learning in social networks and the influence of network structure on decision-making performance of human agents~\cite{WM-DJW:12}. 

%Zhang~\etal~\cite{SL-CC-ZZ:15} study the benefit of communication in the context of the MAB problem with multiple agents. They assume that agents can communicate with every other agent only at certain times and study the effect of communication frequency on the performance of each agent. 
%To the best of our knowledge, the cooperative setting in which agents communicate with only a subset of neighbors has not been explored. 

Here, we use a running consensus algorithm~\cite{braca2008enforcing} for assimilation of information, which is an extension of the classical DeGroot model~\cite{MHDG:74} in the social networks literature.  Running consensus and related models have been used to study learning~\cite{AJ-AS-ATS:10} and decision-making~\cite{VS-NEL:13f} in social networks. 

In the present paper we study the distributed cooperative MAB problem in which a set of agents are faced with a stochastic MAB problem and communicate their information with their neighbors in an undirected, connected communication graph. We use a set of running consensus algorithms for cooperative estimation of the mean reward at each arm, and we design an arm selection heuristic that leads to an order-optimal performance for the group. The major contributions of this paper are as follows.

First, we employ and rigorously analyze running consensus algorithms for distributed cooperative estimation of mean reward at each arm, and we derive bounds { on several relevant quantities. }

Second, we propose and thoroughly analyze the cooperative UCB algorithm. We derive bounds on decision-making performance for the group and characterize the influence of the network structure on the performance. 

Third, we introduce a novel graph centrality measure and numerically demonstrate that this measure captures the ordering of explore-exploit performance of each agent.

The remainder of the paper is organized as follows.  In Section~\ref{Background} we recall some preliminaries about the stochastic MAB problem and consensus algorithms.  In Section~\ref{sec:coop-est} we present and analyze the cooperative estimation algorithm. We propose and analyze the cooperative UCB algorithm in Section~\ref{DistributedDecisionMaking}. We illustrate our analytic results with numerical examples in Section~\ref{NetworkPerformanceAnalysis}. We conclude in Section~\ref{FinalRemarks}.
%use this background to extend the MAB problem to the case of distributed decision-makers communicating over a network through running consensus.  Next, in Section \ref{RegretAnalysis} we demonstrate that the given algorithm results in logarithmic regret.  This is followed in Section \ref{NetworkPerformanceAnalysis} with an analysis of the performance of certain relevant networks under the proposed algorithm, with concluding remarks in Section 

\section{Background}
\label{Background}
In this section we recall the standard MAB problem, the UCB algorithm, and some preliminaries on discrete-time consensus. 

\subsection{The Single Agent MAB Problem}  
%\indent The Multi-armed Bandit problem takes its name from a row of slot machines, which are commonly referred to as one-armed bandits.  Here each arm relates to each machine, where each arm returns a noisy payout.  \\

Consider an $N$-armed bandit problem, i.e., an MAB problem with $N$ arms. 
The reward associated with arm $i\in \until{N}$ is a random variable with an unknown mean $m_i$. 
Let  the agent choose arm $i(t)$ at time $t \in \until{T}$ and receive a reward $r(t)$. 
The decision-maker's objective  is to choose a sequence of arms $\seqdef{i(t)}{t\in \until{T}}$ that maximizes the expected cumulative reward $\sum_{t=1}^T m_{i(t)}$, where $T$ is the horizon length of the sequential allocation process.
%\[
%\underset{i_1, \ldots, i_T}{\maximize} \quad  ,
%\]

For an MAB problem, the expected \emph{regret} at time $t$ is defined by $R(t) = m_{i^*}-m_{i(t)}$, where $m_{i^*} = \max \setdef{m_i}{i\in\until{N}}$. The objective of the decision-maker can be equivalently defined as minimizing the expected cumulative regret defined by $\sum_{t=1}^T R(t) = \sum_{i=1}^N \Delta_i \expt[n_{i}(T)]$, 
where $n_{i}(T)$ is the cumulative number of times arm $i$ has been chosen until time $T$ and $\Delta_i = m_{i^*}-m_i$ is the expected regret due to picking arm $i$ instead of arm $i^*$. It is known that the regret of any algorithm for an MAB problem is asymptotically lower bounded by a logarithmic function of the horizon length $T$~\cite{lai1985asymptotically}, i.e., no algorithm can achieve an expected cumulative regret smaller than a logarithmic function of horizon length as $T \to \infty$. 

In this paper, we focus on Gaussian rewards, i.e., the reward at arm $i$ is sampled from a Gaussian distribution with mean $m_i$ and variance $\sigma^2$.  We assume that the variance $\sigma^2$ is known and is the same at each arm.

\subsection{The UCB Algorithm}

A popular solution to the stochastic MAB problem is the UCB algorithm proposed in~\cite{PA-NCB-PF:02}. The UCB algorithm initializes by sampling each arm once, and then selects an arm with maximum 
\begin{equation*}
Q_i(t) = \hat{\mu}_i(t) + C_i(t), 
\end{equation*}
where $n_i(t)$ is the number of times arm $i$ has been chosen up to and including time $t$, and $\hat{\mu}_i(t)$ and $C_i(t) = \sqrt{\frac{2 \ln(t)}{n_i(t)}}$ are the empirical mean reward of arm $i$ and the associated measure of the uncertainty associated with that mean at time $t$, respectively. 

The function $Q_i(t)$ is judiciously designed to balance the tradeoff between explore and exploit:  the terms $\hat \mu_i(t)$ and $C_i(t)$ facilitate exploitation and exploration, respectively. The UCB algorithm as described above assumes that rewards have a bounded support $[0,1]$, but this algorithm can be easily extended to distributions with unbounded support~\cite{KL-QZ:11}.

%As mentioned above, successful multi-armed bandit algorithms balance exploration and exploitation.  This allows an agent to balance searching for the optimal arm $i^*$, while also gaining reward from what it believes is $i^*$.  The key is to decide on when to perform each of these tasks.  The UCB1 algorithm developed in  accomplishes this heuristically by calculating at each time $t$ and option $i$ a value $Q_i(t)$, where
%
%
%Here  UCB1 prescribes that, after an initial sampling phase where the agent chooses each arm once, the agent choose the arm with $\max\limits_{i=1,...,M}{Q_i(t)}$, and achieves logarithmic regret for rewards with bounded support.

\subsection{The Cooperative MAB Problem}
The cooperative MAB problem is an extension of the single-agent MAB problem where $M$ agents act over the same $N$ arms.  Agents maintain bidirectional communication, and the communication network can be modeled as an undirected graph $\mathcal{G}$ in which each node represents an agent and edges represent the communication between agents~\cite{FB-JC-SM:09}. Let $A\in \real^{M\times M}$ be the adjacency matrix associated with $\mc G$ and let $L \in \real^{M\times M}$ be the corresponding Laplacian matrix. We assume that the graph $\mc G$ is connected, i.e., there exists a path between each pair of nodes.

In the cooperative setting, the objective of the group is defined as minimizing the expected cumulative group regret, defined by $\sum_{k=1}^M \sum_{t=1}^T R^k(t) = \sum_{k=1}^M \sum_{i=1}^N \Delta_i \expt[n_{i}^k(T)]$, 
where $R^k(t)$ is the regret of agent $k$ at time $t$ and $n_{i}^k(T)$ is the total cumulative number of times arm $i$ has been chosen by agent $k$ until time $T$.  In the cooperative setting using Gaussian rewards the lower bound~\cite{VA-PV-JW:87} on the expected number of times a suboptimal arm $i$ is selected by a fusion center that has access  to reward for each agent is
\begin{equation}
\sum_{k=1}^M \expt[n_i^k(T)] \ge \Big( \frac{2 \sigma^2}{\Delta_i^2} +o(1) \Big) \ln T. \label{eqn:fusioncenterregret}
\end{equation}
In the following, we will design a distributed algorithm that samples a suboptimal arm $i$ within a constant factor of the above bound. 

\subsection{Discrete-Time Consensus }

%Consider a set of agents $\until{M}$, each of which maintains bi-directional communication with a set of neighboring agents.  The communication network can be modeled as an undirected graph $\mathcal{G}$ in which each node represents an agent and edges represent the communication between agents~\cite{FB-JC-SM:09}. 
%Let $A\in \real^{M\times M}$ be the adjacency matrix associated with $\mc G$ and let $L \in \real^{M\times M}$ be the corresponding Laplacian matrix. We assume that the graph $\mc G$ is connected, i.e., there exists a path between each pair of nodes. In this setting, it is known~\cite{FB-JC-SM:09} that $L$ is symmetric and positive semidefinite with a one-dimensional nullspace, $\text{span}(\bs 1_M)$, where $\bs 1_M$ is the $M$-column-vector of all ones. 

%For such an undirected  graph, let $\mc N_k$ be the set of neighbors of agent $k$. 
%with a node set $\mathcal{V} = \until{M}$ and an edge set$\mathcal{E}$.  For each agent $k$ we can also define its neighborhood $\mathcal{N}_k= \{k\in\mathcal{V}, (j,k)\in \mathcal{E}\}$. 

%The structure of $\mathcal{G}$ can also be expressed through the adjacency matrix $A$, where $a_{ij}=1$ if $(i,j)\in \mathcal{E}$ for $i \neq j$ and $0$ elsewhere for unweighted graphs.  $\mathcal{G}$ can also be expressed through the Laplacian matrix $L$, where $l_{ij} = -a_{ij}$ for $i\neq j$ and $l_{ii}= \text{deg}(i)$, where $\text{deg}(i)$ is the degree of node $i$.  $L$ is symmetric and positive definite for an undirected graph.    An undirected graph is said to be connected if there is a path from every node to every other node. 

Consider a set of agents $\until{M}$, each of which maintains bidirectional communication with a set of neighboring agents.  The objective of the consensus algorithms is to ensure agreement among agents on a common value. 
%In multi-agent systems a common prerequistite for coordinated action is that agents reach consensus, with averaging being a particular subset \cite{ren2005survey}.  When using averaging consensus each agent maintains an estimate of the average of a state variable, and then updates this variable by performing a weighted average of their state with that of their neighbors.  
%
In the discrete-time consensus algorithm~\cite{AJ-JL-ASM:02, JNT:84}, agents average their opinion with their neighbors' opinions at each time.  A discrete-time consensus algorithm can be expressed as 
\begin{equation}
\mathbf{x}(t) = P\mathbf{x}(t-1),  \label{UpdateEqn}
\end{equation}
where $\mathbf{x}(t)$ is the vector of each agent's opinion, and $P$ is a row stochastic matrix given by
\begin{equation} 
P= \mathcal{I}_M - \frac{\kappa}{d_{\text{max}}} L. \label{Pdefn}
\end{equation}
$\mathcal{I}_M$ is the identity matrix of order $M$, $\kappa \in (0,1]$ is a step size parameter \cite{olfati2004consensus}, $d_{\text{max}} =\max \setdef{\text{deg}(i)}{i \in \until{M}}$, and $\text{deg}(i)$ is the degree of node $i$.  
In the following, we assume without loss of generality that the eigenvalues of $P$ are ordered such that $\lambda_1 = 1 > \lambda_2 \geq ... \geq \lambda_M > -1$. 

In the context of social networks, the consensus algorithm~\eqref{UpdateEqn} is referred to as the Degroot model~\cite{MHDG:74} and has been successfully used to describe evolution of opinions~\cite{BG-MOJ:10}. 

One drawback of the consensus algorithm~\eqref{UpdateEqn} is that it does not allow for incorporating new external information. This drawback can be mitigated by adding a forcing term and the resulting algorithm is called the \emph{running consensus}~\cite{braca2008enforcing}. 
Similar to~\eqref{UpdateEqn}, the running consensus updates the opinion at time $t$ as
\begin{equation}\label{eq:running-consensus}
\mathbf{x}(t) = P\mathbf{x}(t-1) + P \bs \upsilon(t),
\end{equation}
where $\bs \upsilon(t)$ is the information received at time $t$. In the running consensus update~\eqref{eq:running-consensus}, each agent $k$ collects information $\upsilon_k(t)$ at time $t$, adds it to its current opinion, and then averages its updated opinion with the updated opinion of its neighbors.

%This construction results in a symmetric, doubly stochastic $P$ for undirected graphs.  Additionally, the eigenvalues of P, denoted $\lambda_p$ for $p=1,...,M$, are all within the unit circle and real, with one eigenvalue at $1$.  
%Here we have ordered the eigenvalues such that $\lambda_1 = 1 > \lambda_2 \geq ... > \lambda_M > -1$. If $\mathcal{G}$ is undirected and connected, a multi-agent system using \eqref{UpdateEqn} and \eqref{Pdefn} will settle on a final equilibrium position where each agent's state is the average of every agent's initial condition \cite{olfati2004consensus}.  In general, graphs that are more ``connected'' will converge to equilibrium faster, and precise bounds can be computed that depend on $\lambda_2$ of $P$, the ``Fiedler'' eigenvalue \cite{braca2008enforcing}.

\section{Cooperative Estimation of Mean Rewards} \label{sec:coop-est}
In this section we investigate the cooperative estimation of mean rewards at each 
arm. To this end, we propose two running consensus algorithms for each arm and analyze their performance.

\subsection{Cooperative Estimation Algorithm}
For distributed cooperative estimation of the mean reward at each arm $i$, we employ two running consensus algorithms: (i) for estimation of total reward provided at the arm, and (ii) for estimation of the total number of times the arm has been sampled. 

Let $\hat{s}_i^k(t)$ and $\hat{n}_i^k(t)$ be agent $k$'s estimate of the total reward provided at arm $i$ per unit agent and the total number of times arm $i$ has been selected until time $t$ per unit agent, respectively. Using $\hat{s}_i^k(t) $ and  $\hat{n}_i^k(t) $ agent $k$ can calculate $\hat{\mu}_i^k(t)$, the estimated empirical mean of arm $i$ at time $t$ defined by
\begin{equation}
\hat{ \mu}_i^{k}(t) = \frac{\hat{s}_i^{k}(t)}{\hat{n}_i^{k}(t)}.  \label{eqnmean}
\end{equation}

%We assume that the total number of agents $M$ is known to every agent. Thus, each agent can estimate the total reward received across agents at arm $i$ and the  number of times arm $i$ has been visited by all agents using  $\hat{s}_i^k(t)$ and $\hat{n}_i^k(t)$. 

%We assume that each agent will estimate the sum of rewards $\hat{s}_i^k(t)$  for arm $i$ and the sum of the the number of times arm $i$ has been selected $\hat{n}_i^k(t)$.  
%
%Here $M\hat{s}_i^k(t)$ is agent $k$'s estimate of the total reward received by all $M$ agents at arm $i$ up until time $t$, and $M\hat{n}_i^k(t)$ is agent $k$'s estimate of the number of times arm $i$ has been chosen up until time $t$ by all $M$ agents.  

%Let $\mathbf{\hat{s}}_i(t) =  \lbrack \hat{s}_i^1(t),...,\hat{s}_i^M(t)\rbrack^\top$, and $\mathbf{\hat{n}}_i(t) =  \lbrack \hat{n}_i^1(t),...,\hat{n}_i^M(t)\rbrack^\top$.  

Let $i^k(t)$ be the arm sampled by agent $k$ at time $t$ and let $\xi_i^k(t) = \mathds{1} (i^k(t) =i)$.   $\mathds{1}(\cdot)$ is the indicator function, here equal to 1 if $i^k(t)= i$ and 0 otherwise. 
{For simplicity of notation we define $r_i^k(t)$ as the realized reward at arm $i$ for agent $k$, which is a random variable sampled from $\mathcal{N}(m_i,\sigma^2)$, and the corresponding accumulated reward is $r^k(t) = r_i^k(t) \cdot \mathds{1} (i^k(t) =i)$.}

The estimates $\hat{s}_i^k(t)$ and $\hat{n}_i^k(t)$ are updated using running consensus as follows
\begin{align} \label{nhatdefn}
\mathbf{\hat{n}}_i(t) &= P \mathbf{\hat{n}}_i(t-1) + P\boldsymbol{\xi}_i(t), \\
\text{and} \quad \mathbf{\hat{s}}_i(t) &= P \mathbf{\hat{s}}_i(t-1) + { P(\mathbf{r}_i(t) \circ \boldsymbol{\xi}_i(t)),} \label{shatdefn}
\end{align}
where $\mathbf{\hat{ n}}_i(t)$, $\mathbf{\hat{ s}}_i(t)$, $\boldsymbol{\xi}_i(t)$, and $\mathbf{r}_i(t)$ are vectors of $\hat n_i^k(t)$, $\hat s_i^k(t)$, $\xi_i^k(t)$, and $r_i^k(t)$, $k\in \until{M}$, respectively, { and $\circ$ denotes element-wise multiplication (Hadamard product).}

%$\boldsymbol{\xi}_i(t) = \lbrack \xi_i^1(t),...,\xi_i^M(t)\rbrack^\top$ 

%In the version of running consensus used here, each agent $k$  maintains a consensus value  for each arm $i$ of $\hat{n}_i^k(t)$ and $\hat{s}_i^k(t)$.   At each timestep each agent uses $\hat{n}_i^k(t)$ and $\hat{s}_i^k(t)$ to make an arm choice, and receives a reward.  Agent $k$ then adds the reward to $\hat{s}_i^k(t)$, and $1$ to $\hat{n}_i^k(t)$ if arm $i$ was chosen, and then performs consensus with neighbors.  This is the same as the running consensus algorithm in \cite{braca2008enforcing}, but the new values are added in prior to consensus, and not during.  This change allows for considerably greater certainty in an agent's estimates with little additional computational complexity, and this results to considerably better performance.  Additionally, here we are averaging running sums, whereas \cite{braca2008enforcing} was averaging means and thus used different weights.\\
%\indent The running consensus algorithm used here can be expressed as
%
%
%Here and $\mathbf{r}_i(t) = \lbrack r_i^1(t),...,r_i^M(t)\rbrack^\top$, where $r_i^k(t) \sim \mathcal{N}(m_i,\sigma) $ is the reward received by agent $k$ from arm $i$ at time $t$.  Given that $P$ is formulated as in \eqref{Pdefn} and $\mathcal{G}$ is connected $\hat{n}_i^k(t)$ and $\hat{s}_i^k(t)$ will converge to the correct average consensus value.  

\subsection{Analysis of the Cooperative Estimation Algorithm}
We now analyze the performance of the estimation algorithm defined by~\eqref{eqnmean},~\eqref{nhatdefn}~and~\eqref{shatdefn}. Let $n_i^{\text{cent}}(t) \equiv \frac{1}{M} \sum_{\tau=1}^t \mathbf{1}_M^\top \boldsymbol{\xi}_i(\tau)$ be the total number of times arm $i$ has been selected per unit agent up to and including time $t$, and let $s_i^{\text{cent}}(t) \equiv \frac{1}{M} \sum_{\tau=1}^t \boldsymbol{\xi}_i^\top(t) \mathbf{r}_i(t)$ be the
total reward provided at arm $i$ per unit agent up to and including time $t$.  Also, let $\lambda_i$ denote the $i$-th largest eigenvalue of $P$, $\mathbf{u}_i$  the eigenvector corresponding to $\lambda_i$, $u_i^d$  the $d$-th entry of $\mathbf{u}_i$, and
\begin{equation}
\epsilon_n =  \sqrt{M}  \sum_{p=2}^M \frac{|\lambda_p|}{1-|\lambda_p|}.  \label{epsilonndef}
\end{equation}
Note that $\lambda_1=1$ and $\mathbf{u}_1 = \mathbf{1}_M/\sqrt{M}$. Let us define
\begin{align*}
 \nu_{pj}^{\text{+sum}} &= \sum_{d=1}^M  u_p^d u_j^d \mathds{1}(u_p^k u_j^k \geq 0) \\
 \text{and} \quad \nu_{pj}^{\text{-sum}} &=  \sum_{d=1}^M u_p^d u_j^d \mathds{1}(u_p^k u_j^k \leq 0). 
\end{align*}
%and let $\nu_{pj}(\tau) \in [ \nu_{pj}^{\text{-sum}},  \nu_{pj}^{\text{+sum}}]$.  
We also define
\begin{equation}
a_{pj}(k) =  \begin{cases}
\nu_{pj}^{\text{+sum}}u_p^k u_j^k, & 
\! \! \! \! \text{if } \lambda_p \lambda_j \geq 0 \; \&\; u_p^k u_j^k  \geq 0, \\
\nu_{pj}^{\text{-sum}}u_p^k u_j^k, &
\! \! \! \! \text{if }  \lambda_p \lambda_j \geq 0 \; \&\; u_p^k u_j^k  \leq 0, \\
\nu_{pj}^{\text{max}} |u_p^k u_j^k | , & \! \! \! \! \text{if }  \lambda_p \lambda_j < 0,
\end{cases}\label{apjdefn}
\end{equation}
where $\nu_{pj}^{\text{max}} = \max{\{ |\nu_{pj}^{\text{-sum}}|,  \nu_{pj}^{\text{+sum}}\}} $. Furthermore, let 
\begin{equation}\label{eq:epsilon-c}
\epsilon_c^k   =    M \sum_{p=1}^M \sum_{j=2}^M \frac{|\lambda_p \lambda_j| }{1-|\lambda_p \lambda_j|} a_{pj}(k).
\end{equation}

We note that both $\epsilon_n$ and $\epsilon_c^k$ depend only on the topology of the communication graph.  These are measures of distributed cooperative estimation performance.

%We investigate the convergence of $\mathbf{\hat{n}}_i(t)$, the convergence of the expected value of $\mathbf{\hat{s}}_i(t)$, denoted by $\mathbb{E}[\mathbf{\hat{s}}_i(t)]$, and the covariance of $\mathbf{\hat{s}}_i(t)$ as a function of the network structure. 

%\begin{definition}\label{ncentdefn}
%The value of $n_i(t)$ for a centralized agent is:
%\begin{equation*}
%n_i^{\text{cent}}(t) = \frac{1}{M} \sum_{t=0}^\top \mathbf{1}_M^\top \boldsymbol{\xi}_i(t)
%\end{equation*}

%\end{definition}
%\bigskip

%\begin{definition}\label{scentdefn}
%The value of $s_i(t)$ for a centralized agent is:
%\begin{equation*}
%
%\end{equation*}

%\end{definition}

\begin{proposition}[\bit{Performance of cooperative estimation}]\label{prop:coop-est}
For the distributed estimation algorithm defined in~\eqref{eqnmean},~\eqref{nhatdefn}~and~\eqref{shatdefn}, and a doubly stochastic matrix $P$ defined in~\eqref{Pdefn}, the following statements hold
\begin{enumerate}
\item the estimate $\hat n_i^k(t)$ satisfies
\begin{align*}
n_i^{\text{cent}}(t) - \epsilon_n \le  \hat{n}_i^k(t) \le   n_i^{\text{cent}}(t) + \epsilon_n;
\end{align*}
\item { the following inequality holds for the estimate   $\hat {n}_i^k(t)$ and the sequence $\seqdef{\xi_i^j(\tau)}{\tau\in \until{t}}$,$ j \in \until{M}$
\begin{equation*}
\sum_{\tau=1}^{t} \sum_{j=1}^M \left(\sum_{p=1}^M \lambda_p^{t-\tau+1}  u_p^k u_p^j  \right)^2 \xi_i^j(\tau) \leq \frac{ \hat n_i^k(t) + \epsilon_c^k}{M}.
\end{equation*}}
%the variance of the estimate $\hat s_i^k(t)$ satisfies
%\begin{equation*}
%\text{Var}[{\hat{s}}_i^k (t)] \leq  \frac{\sigma^2 }{M}(\hat n_i^k(t) + \epsilon_c^k).
%\end{equation*}
\end{enumerate}
\end{proposition}
\begin{proof}
We begin with the first statement. 
From~\eqref{nhatdefn} it follows that
%We begin by writing the dynamics for $\mathbf{\hat{n}}_i(t)$:
%$$
%\mathbf{\hat{n}}_i(t+1) = P \mathbf{\hat{n}}_i(t) + P\boldsymbol{\xi}_i(t).
%$$
%From this, we can write down the solution for $\mathbf{\hat{n}}_i(t)$ at all future times as
\begin{align}
\mathbf{\hat{n}}_i(t) & =P^t \mathbf{\hat{n}}_i(0) + \sum_{\tau = 1}^{t} P^{t-\tau+1} \boldsymbol{\xi}_i(\tau) \nonumber \\
&=\sum_{\tau = 0}^{t}  \Big[ \frac{1}{M} \mathbf{1}_M \mathbf{1}_M^\top \boldsymbol{\xi}_i(\tau) + \sum_{p=2}^{M} \lambda_p^{t-\tau+1} \mathbf{u}_p \mathbf{u}_p^\top \boldsymbol{\xi}_i(\tau) \Big] \nonumber\\
&=  n_i^{\text{cent}}(t) \mathbf{1}_M + \sum_{\tau = 1}^{t} \sum_{p=2}^{M} \lambda_p^{t-\tau+1} \mathbf{u}_p \mathbf{u}_p^\top \boldsymbol{\xi}_i(\tau). \label{solndecomp}
\end{align}
%and further expand the terms through modal decomposition as
%\begin{align*}
%\mathbf{\hat{n}}_i(t) &= \sum_{\tau = 0}^{t-1} P^{t-\tau} \boldsymbol{\xi}_i(\tau) \\
%&=\sum_{\tau = 0}^{t-1}    (U \Lambda U^\top)^{t-\tau} \boldsymbol{\xi}_i(\tau).
%\end{align*}
%\indent Noting Definition \ref{ncentdefn} and that $\lambda_1=1$ leads to
%\begin{align}
%%\mathbf{\hat{n}}_i(t)
%%&\nonumber \\
%&= \mathbf{1}_M n_i^{\text{cent}}(t) + \sum_{\tau = 0}^{t-1} \sum_{p=2}^{M} \lambda_p^{t-\tau} \mathbf{u}_p {\mathbf{u}_p}^\top \boldsymbol{\xi}_i(\tau) \nonumber \\
%
%\end{align}
We now bound the $k$-th entry of the second term on the right hand side of~\eqref{solndecomp}:
%This shows that $\hat{n}_i^k(t)$ will always be close to ${n}_i^{\text{cent}}(t)$ by a certain margin as given by the second term of \eqref{solndecomp}.  The challenge is now to bound that term in order to obtain a bound on $\hat{n}_i^k(t)$, which results in
\begin{align}
\sum_{\tau = 1}^{t} \sum_{p=2}^M  \lambda_p^{t-\tau+1}  \big( \mathbf{u}_p \mathbf{u}_p^\top \boldsymbol{\xi}_i(\tau) \big)_k  \!
% \leq \Big\| \sum_{\tau = 1}^{t} \sum_{p=2}^M \lambda_p^{t-\tau} \mathbf{u}_p {\mathbf{u}_p}^\top \boldsymbol{\xi}_i(\tau) \|_2 \nonumber \\ 
&\leq \!   \sum_{\tau = 1}^{t} \sum_{p=2}^M |\lambda_p^{t-\tau+1}| \| \mathbf{u}_p\|_2^2   \| \boldsymbol{\xi}_i(\tau) \|_2  \nonumber \\
& \leq   \sqrt{M} \sum_{\tau = 1}^{t} \sum_{p=2}^M  |\lambda_p^{t-\tau+1}| \le \epsilon_n.  \nonumber
%\\
%%& =   \sqrt{M}   \left(  \sum_{p=2}^M  \frac{|\lambda_p|(1-|\lambda_p^{t-1}|)}{1-|\lambda_p|}  \right)  \nonumber  \\
%& \leq   \sqrt{M}  \sum_{p=2}^M \frac{|\lambda_p|}{1-|\lambda_p|}  \label{3rdtermbound} \\
%& \equiv \epsilon_n \label{epsilonn}.
\end{align}
This establishes the first statement. 

%Start large changed section
{
To prove the second statement, we note that 
\begin{align}
\sum_{\tau=1}^{t} &\sum_{j=1}^M \left(\sum_{p=1}^M \lambda_p^{t-\tau+1}  u_p^k u_p^j  \right)^2 \xi_i^j(\tau) \nonumber \\
& = \sum_{\tau=1}^t \sum_{p=1}^M \sum_{w=1}^M (\lambda_p \lambda_w)^{t-\tau+1} u_p^k u_w^k \sum_{j=1}^M u_p^j u_w^j \xi_i^j(\tau) \nonumber \\
& = \sum_{\tau=1}^t \sum_{p=1}^M \sum_{w=2}^M (\lambda_p \lambda_w)^{t-\tau+1} u_p^k u_w^k \nu_{pwi}(\tau) \nonumber \\
& \quad \quad + \frac{1}{M} \sum_{\tau=1}^t \sum_{p=1}^M  \sum_{j=1}^M \lambda_p^{t-\tau+1} u_p^k u_p^j \xi_i^j(\tau) \nonumber \\
& = \sum_{\tau=1}^t \sum_{p=1}^M \sum_{w=2}^M (\lambda_p \lambda_w)^{t-\tau+1} u_p^k u_w^k \nu_{pwi}(\tau) +\frac{1}{M} \hat{n}_i^k(t), \label{covssoln}
\end{align}
where $\nu_{pwi}(\tau) \!=\! \sum_{j=1}^M u_p^j u_w^j \xi_i^j(\tau)$.

We now analyze the first term of \eqref{covssoln}:
%We conclude the bounds with the $kk$'th entry of  As seen later, it is vital to make this bound as tight as possible in order to maximize performance, which can be done as follows, beginning with
\begin{align}
& \sum_{\tau=1}^t \sum_{p=1}^M \sum_{w=2}^M (\lambda_p \lambda_w)^{t-\tau+1} u_p^k u_w^k \nu_{pwi}(\tau) \nonumber \\ 
&  \leq \sum_{\tau=1}^t \sum_{p=1}^M \sum_{w=2}^M |(\lambda_p \lambda_w)^{t-\tau+1} || u_p^k u_w^k \nu_{pwi}(\tau) | \nonumber \\ 
& \leq  \sum_{\tau = 0}^{t-1}  \sum_{p=1}^M \sum_{w=2}^M |\lambda_p \lambda_w|^{t-\tau+1} a_{pw}(k) \nonumber \\ 
&\leq   \sum_{p=1}^M \sum_{w=2}^M \frac{|\lambda_p \lambda_w|}{1-|\lambda_p \lambda_w|}  a_{pw}(k). \label{cov1stmid}
\end{align}
%Plugging back into \eqref{cov1stmid} results in 
%\begin{align} 
%\sum_{\tau = 0}^{t-1} & \sum_{p=1}^M \sum_{j=2}^M |(\lambda_p \lambda_j)^{t-\tau} | | \nu_{pj}(\tau) (\mathbf{u}_p \mathbf{u}_j^\top)_{kk} |  \nonumber \\
%& \leq \sum_{\tau = 0}^{t-1}  \sum_{p=1}^M \sum_{j=2}^M |(\lambda_p \lambda_j)^{t-\tau} | a_{pj}(k) \nonumber \\
%& \leq   \sum_{p=1}^M \sum_{j=2}^M \frac{|\lambda_p \lambda_j|(1-|(\lambda_p \lambda_j)^{t-1}|)}{1-|\lambda_p \lambda_j|}  a_{pj}(k) \nonumber \\
%& \leq   \sum_{p=1}^M \sum_{j=2}^M \frac{|\lambda_p \lambda_j|}{1-|\lambda_p \lambda_j|}  a_{pj}(k) \nonumber 
%\end{align}
Bounds in \eqref{cov1stmid} establish the second statement. 
}
\end{proof}

{
We now  derive bounds on the deviation of the estimated mean when using the cooperative estimation algorithm.  {We use techniques from \cite{AG-EM:08}.}  Recall that for $i\in \{1,\dots,N\}$ and $k\in \{1,\dots,M\}$ {we} let $\seqdef{r_i^k(t)}{t \in \naturals}$  be the sequence of i.i.d. Gaussian random variables with mean $m_i \in \real$.   Let $\mathcal{F}_t$ be the filtration defined by the sigma-algebra of all the measurements until time $t$. 
Let $\seqdef{\xi_i^k(t)}{t \in \naturals}$  be a sequence of Bernoulli variables such that $\xi_i^k(t)$ is deterministically known given $\mc F_{t-1}$, i.e., $\xi_i^k(t)$ is pre-visible w.r.t. $\mc F_{t-1}$.
Additionally, let $\phi_i(\beta) = \lnn{\mathbb{E}[\expp{\beta r_i^k(t)}]}$ denote the cumulant generating function of $r_i^k(t)$.  
%, that is an increasing sequence of $\sigma-\text{algebras}$ of $\mathcal{A}$ such that for every $\sigma(r_i^k(1),\dots,r_i^k(t))\subset \mathcal{F}_t$ and for every $\tau > t$ $r_i^k(\tau)$ is independent of $\mathcal{F}_t$.  

\begin{theorem}[\bit{Estimator Deviation Bounds}] \label{Thm:EstDevBoudnsCondensed}
	For the estimates 	$\hat{s}_i^k(t)$ and $\hat{n}_i^k(t)$ obtained using equations~\eqref{nhatdefn} and \eqref{shatdefn}, the following concentration inequality holds 
	\begin{equation}
		\mathbb{P}\Bigg( \! \frac{\hat{s}_i^k(t) \!-\! m_i \hat{n}_i^k(t)}{(\frac{1}{M} \left(\hat{n}_i^k(t) \!+\! \epsilon_c^k\right))^{\slfrac{1}{2}}} \! > \! \delta \! \Bigg) \!<\! \Bigg\lceil \!\frac{\lnn{t \!+\! \epsilon_n}}{\lnn{1\!+\!\eta}}\! \Bigg\rceil  \!\expp{\frac{-\delta^2}{2\sigma^2} G(\eta) \! \! } ,
	\end{equation}
	where 	$\delta > 0$, $\eta >0$, $G(\eta) = (1-\frac{\eta^2}{16})$, and $\epsilon_c^k$ and $\epsilon_n$ are defined in~\eqref{eq:epsilon-c} and~\eqref{epsilonndef}, respectively. 
\end{theorem}
\begin{proof}
	We begin by noting that $\hat{s}_i^k(t)$ can be decomposed as 
	\begin{equation}
		\hat{s}_i^k(t) = \sum_{\tau=1}^{t} \sum_{p=1}^{M} \lambda_p^{t-\tau+1} \sum_{j=1}^M u_p^k u_p^j r_i^j(\tau) \xi_i^j(\tau). \label{eqn:s_i_hat_defn}
	\end{equation}
	Let $\hat{s}_i^{kp}(t) =  \sum_{\tau=1}^{t} \lambda_p^{t-\tau+1} \sum_{j=1}^M u_p^k u_p^j r_i^j(\tau) \xi_i^j(\tau)$.  Then, 
	\begin{align}
		\sum_{p=1}^M \hat{s}_i^{kp}(t)  =   \sum_{p=1}^{M} \sum_{j=1}^M \lambda_p u_p^k u_p^j r_i^j(t) \xi_i^j(t) + \sum_{p=1}^M \lambda_p \hat{s}_i^{kp}(t-1). \label{eqn:shat_modaldecomp}
	\end{align}
	%	\begin{align}
	%	\hat{s}_i^k(t) &= \sum_{\tau=1}^{t} \sum_{p=1}^{M} \lambda_p^{t-\tau+1} \sum_{j=1}^M u_p^k u_p^j r_i^j(\tau) \xi_i^j(\tau). \nonumber\\
	%	& = \lambda_p  \sum_{p=1}^{M} \sum_{j=1}^M u_p^k u_p^j r_i^j(t) \xi_i^j(t) + \lambda_p \hat{s}_i^k(t-1). \label{eqn:shat_modaldecomp}
	%	\end{align}
	%	Additionally, note that for any $\Theta > 0$ that 
	%	\begin{align*}
	%	\mathbb{E}&\left[ \expp{\Theta  \sum_{p=1}^{M} \lambda_p \sum_{j=1}^M u_p^k u_p^j r_i^j(t) \xi_i^k(t)} \middle \vert F_{t-1} \right] \\
	%	& = \expp{\sum_{j=1}^M \phi_i^{t} \left( \Theta \sum_{p=1}^M \lambda_p u_p^k u_p^j r_i^j(t)  \right) \xi_i^j(t) }
	%	\end{align*}
	%	where the previous equality holds because $\xi_i^j(t)$ is predictable for every $\tau < t$ and $r_i^j(t)$ is independent from $\mathcal{F}_{t-1}$.  
	%
	It follows from \eqref{eqn:s_i_hat_defn} and \eqref{eqn:shat_modaldecomp} that for any $\Theta > 0$
	\begin{align*}
		\mathbb{E}& \left[\expp{\Theta \hat{s}_i^k(t)} \middle \vert \mc F_{t-1} \right] =
		\mathbb{E}\left[\expp{\Theta \sum_{p=1}^M \hat{s}_i^{kp}(t)} \middle \vert \mc F_{t-1} \right] 	
		\\
		& = \mathbb{E} \left[ \expp{\Theta \sum_{p=1}^M \lambda_p \sum_{j=1}^M u_p^k u_p^j r_i^j(t) \xi_i^j(t)}  \middle \vert \mc F_{t-1} \right] \\
		& \qquad \qquad \times \expp{\Theta \sum_{p=1}^M \lambda_p \hat{s}_i^{kp}(t-1)}\\
		& = \expp{\sum_{j=1}^M \phi_i \left( \Theta \sum_{p=1}^M \lambda_p u_p^k u_p^j r_i^j(t)  \right) \xi_i^j(t) } \\
		& \qquad \qquad \times \expp{\Theta \sum_{p=1}^M \lambda_p \hat{s}_i^{kp}(t-1)}, 
		%	& \qquad \qquad \cdot \expp{\Theta \sum_{p=1}^M \sum_{\tau=1}^{t-1} \sum_{j=1}^M \lambda_p^{t-\tau+1} u_p^k u_p^j r_i^j(\tau) \xi_i^j(\tau)}
	\end{align*}
	where the last line follows using the fact that conditioned on $\mc F_{t-1}$, $ \xi_i^j(t)$ is a deterministic variable and $r_i^j(t)$ are i.i.d. for each $j \in \until{M}$. Therefore, it follows that 
	\begin{multline*}
		\mathbb{E} \left[\exp \bigg(\Theta \sum_{p=1}^M \hat{s}_i^{kp}(t) - \sum_{j=1}^M \phi_i \left( \Theta \sum_{p=1}^M \lambda_p u_p^k u_p^j r_i^j(t)  \right)\right. \\ 
		\times	 \xi_i^j(t)\bigg)
		\left.	\bigg| \mc F_{t-1} \right] 
		= \expp{\Theta \sum_{p=1}^M \lambda_p \hat{s}_i^{kp}(t-1)}. 
	\end{multline*}
	Using the above argument recursively with the fact that $ s_i^k(0)=0$, we obtain
	%	from statement 2 of Proposition 1. Therefore, as $\phi_i^0 (\cdot) = 0$, by induction this proves that 
	%	\begin{align*}
	%	1& = \mathbb{E} \Bigg[\exp \Bigg( \Theta \sum_{p=1}^M \sum_{\tau=1}^t \sum_{j=1}^M \lambda_p^{t-\tau+1} u_p^k u_p^j r_i^j(\tau) \xi_i^j(\tau)  \\
	%	& \qquad - \sum_{\tau=1}^t \sum_{j=1}^M \phi_i^\tau \left(\Theta \sum_{p=1}^M \lambda_p^{t-\tau+1} u_p^k u_p^j r_i^j(\tau) \right) \xi_i^j(\tau) \Bigg ) \Bigg].
	%	\end{align*}
	\begin{multline*}
		\mathbb{E} \Bigg[\exp \Bigg( \Theta  \hat{s}_i^k(t)  \\
		- \sum_{\tau=1}^t \sum_{j=1}^M \phi_i \left(\Theta \sum_{p=1}^M \lambda_p^{t-\tau+1} u_p^k u_p^j r_i^j(\tau) \right) 
		\xi_i^j(\tau) \Bigg ) \Bigg] =1. 
	\end{multline*}
	Since for 	Gaussian random variables $\phi_i(\beta) = \beta m_i + \frac{1}{2}\sigma^2 \beta^2$, we have
	\begin{align*}
		1 &=\mathbb{E} \Bigg [ \! \exp \Bigg( \Theta \!\left(\hat{s}_i^k\!(t) \!-\! m_i \hat{n}_i^k\!(t)\right) \\
		& \qquad \qquad - \frac{\sigma^2}{2}\! \sum_{\tau=1}^{t} \sum_{j=1}^M \!\left(\! \Theta\! \sum_{p=1}^M \lambda_p^{t-\tau+1} \! u_p^k u_p^j  \right)^2  \! \! \xi_i^j(\tau) \Bigg) \Bigg ] \\
		&\ge \mathbb{E} \Bigg [ \expp{ \! \Theta \!\left(\hat{s}_i^k\!(t) \!-\! m_i \hat{n}_i^k\!(t)\right) \!-\! \frac{\sigma^2 \Theta^2}{2M} \left(\hat{n}_i^k(t) + \epsilon_c^k \right)}\Bigg ],
	\end{align*}
	where the last inequality follows from the second statement of Proposition~\ref{prop:coop-est}. 	
	Now using the Markov Inequality, we obtain
	\begin{align}
		&e^{-a} \! \geq \mathbb{P} \! \left( \! \expp{ \! \Theta \!\left(\hat{s}_i^k\!(t) \!-\! m_i \hat{n}_i^k\!(t)\right) \!-\! \frac{\sigma^2 \Theta^2}{2M} \left(\hat{n}_i^k(t) + \epsilon_c^k \right)} \! \! \geq \! e^{a} \! \right)  \nonumber\\
		&= \mathbb{P} \Bigg( \frac{\hat{s}_i^k(t) - m_i \hat{n}_i^k(t)}{\left( \frac{1}{M}\left(\hat{n}_i^k(t) + \epsilon_c^k \right)  \right)^{\frac{1}{2}}}  \geq  \frac{a}{\Theta} \left( \frac{1}{M}\left(\hat{n}_i^k(t) + \epsilon_c^k \right)  \right)^{-\frac{1}{2}} \nonumber \\
		& \qquad \qquad \qquad \qquad  + \frac{\sigma^2 \Theta}{2}\left( \frac{1}{M}\left(\hat{n}_i^k(t) + \epsilon_c^k \right)  \right)^{\frac{1}{2}} \Bigg ). \label{eq:ineq1}
	\end{align}
	The right hand side of the above {inequality} contains a random variable $\hat n_i^k(t)$ which is dependent on the random variable on the left hand side. 	Therefore, we use union bounds on $\hat n_i^k(t)$ to obtain the desired concentration inequality. Towards this end, we consider an exponentially increasing sequence of time indices $\setdef{(1+\eta)^{g-1}}{g \in \until{D}}$, where $D = \left \lceil \frac{\lnn{t+ \epsilon_n}}{\lnn{1 + \eta}} \right \rceil$ and $\eta >0$. 
	%	
	%	Now, let $D = \left \lceil \frac{\lnn{t+ \epsilon_n}}{\lnn{1 + \eta}} \right \rceil$ for $\eta > 0$, and
	For every $g \in \{1,\dots,D\}$, define
	\begin{equation}
		\Theta_g = \frac{1}{\sigma}  \sqrt{\frac{2a M}{(1+\eta)^{g-\frac{1}{2}} + \epsilon_c^k}}.
	\end{equation}
	Thus, if $(1+\eta)^{g-1} \leq \hat{n}_i^k(t) \leq (1+\eta)^{g}$, then
	\begin{align}
		&\frac{a}{\Theta_g} \left( \frac{1}{M}\left(\hat{n}_i^k(t) + \epsilon_c^k \right)  \right)^{-\frac{1}{2}}  + \frac{\sigma^2 \Theta_g}{2}\left( \frac{1}{M}\left(\hat{n}_i^k(t) + \epsilon_c^k \right)  \right)^{\frac{1}{2}} \nonumber\\
		& = \sigma \sqrt{\frac{a}{2}} \left(\! \left( \frac{(1+ \eta)^{g-\frac{1}{2}} + \epsilon_c^k }{\hat{n}_i^k(t) + \epsilon_c^k} \right)^{\frac{1}{2}} \!\! + \left( \frac{\hat{n}_i^k(t) + \epsilon_c^k}{(1+ \eta)^{g-\frac{1}{2}} + \epsilon_c^k } \right)^{\frac{1}{2}}   \! \right) \nonumber \\
		& \le \sigma \sqrt{\frac{a}{2}} \left(\! \left( \frac{(1+ \eta)^{g-\frac{1}{2}} }{\hat{n}_i^k(t) } \right)^{\frac{1}{2}} \!\! + \left( \frac{\hat{n}_i^k(t) }{(1+ \eta)^{g-\frac{1}{2}}} \right)^{\frac{1}{2}}   \! \right) \nonumber \\
		&\leq \sigma \sqrt{\frac{a}{2}} \left(\! (1+\eta)^{\frac{1}{4}} +  (1+\eta)^{-\frac{1}{4}}   \! \right),  \label{eq:ineq2}
	\end{align}
	where the second-to-last inequality follows from the fact that for $a,b >0$, the function $\epsilon \mapsto \sqrt{\frac{a + \epsilon}{b+\epsilon}} + \sqrt{\frac{b + \epsilon}{a+\epsilon}}$ with domain $\real_{\ge 0}$ is monotonically non-increasing, and
	the last inequality follows from the fact the for $\eta>0$, the function $ x \mapsto \sqrt{\frac{(1+ \eta)^{g-\frac{1}{2}}}{x}} \ + \sqrt{\frac{x}{(1+ \eta)^{g-\frac{1}{2}}}} $ with domain $ [(1+\eta)^{g-1}, (1+\eta)^g]$ achieves its maximum at either of the boundaries. 
	%where the last inequality follows because the function in the second-last line is non-decreasing in $\epsilon_c^k$. {\color{red} Why?}
	%	
	%	where the previous inequality holds for all $z$ such that $(1+\eta)^{g-1} \leq z \leq (1+\eta)^{g}$, $\eta>0$, and $\epsilon>0$ because
	%	\begin{align*}
	%	& \left( \frac{(1+\eta)^{g-\frac{1}{2}} + \epsilon}{z+ \epsilon} \right)^{\frac{1}{2}}  +  \left( \frac{z+ \epsilon}{(1+\eta)^{g-\frac{1}{2}} + \epsilon} \right)^{\frac{1}{2}} \\
	%	& \qquad\leq  \left( \frac{(1+\eta)^{g-\frac{1}{2}}}{z} \right)^{\frac{1}{2}}  +   \left( \frac{z}{(1+\eta)^{g-\frac{1}{2}}} \right)^{\frac{1}{2}} \\
	%	& \qquad \leq (1+\eta)^{\frac{1}{4}}  +  (1+\eta)^{-\frac{1}{4}}
	%	\end{align*}
	%	Therefore,
	%	\begin{align*}
	%	&\Bigg\{ \frac{\hat{s}_i^k(t) - m_i \hat{n}_i^k(t)}{\left(\frac{1}{M} \left(\hat{n}_i^k(t) + \epsilon_c^k\right)\right)^{\frac{1}{2}}} > \sigma \sqrt{\frac{a}{2}} \left( (1+\eta)^{\frac{1}{4}} + (1+\eta)^{-\frac{1}{4}}   \right)  \Bigg\}\\
	%	& \subset \bigcup_{g=1}^D \Bigg \{ \frac{\hat{s}_i^k(t) - m_i \hat{n}_i^k(t)}{\left(\frac{1}{M} \left(\hat{n}_i^k(t) + \epsilon_c^k\right)\right)^{\frac{1}{2}}} > \frac{a}{\Theta_k}\left(\frac{1}{M} \left(\hat{n}_i^k(t) + \epsilon_c^k\right)\right)^{-\frac{1}{2}} \\
	%	& \qquad \qquad \qquad \qquad  + \frac{\sigma^2 \Theta_k}{2} \left(\frac{1}{M} \left(\hat{n}_i^k(t) + \epsilon_c^k\right)\right)^{\frac{1}{2}}  \Bigg \}.
	%	\end{align*}
	%Using the union bound implies
	
{	
It follows from~\eqref{eq:ineq1} that	
	\begin{align*}
		&e^{-a} \! \geq \mathbb{P} \Bigg( \frac{\hat{s}_i^k(t) - m_i \hat{n}_i^k(t)}{\left(\frac{1}{M} \left(\hat{n}_i^k(t) + \epsilon_c^k\right)\right)^{\frac{1}{2}}} \! > \! \frac{a}{\Theta_g}\left(\frac{1}{M} \left(\hat{n}_i^k(t) + \epsilon_c^k\right)\right)^{-\frac{1}{2}}  \\
		&\qquad \qquad +\frac{\sigma^2 \Theta_g}{2} \left(\frac{1}{M} \left(\hat{n}_i^k(t) + \epsilon_c^k\right)\right)^{\frac{1}{2}}\\
		&\qquad \quad   \; \&\; (1+\eta)^{g-1} \le \hat n_i^k(t) + \epsilon_c^k< (1+\eta)^g   \Bigg),
	\end{align*}
for any $g \in \until{D}$. 	
	Therefore,
	\begin{align*}
&D e^{-a}	\ge \sum_{g=1}^D \mathbb{P} \Bigg( \frac{\hat{s}_i^k(t) - m_i \hat{n}_i^k(t)}{\left(\frac{1}{M} \left(\hat{n}_i^k(t) + \epsilon_c^k\right)\right)^{\frac{1}{2}}} \! > \! \frac{a}{\Theta_g}\left(\frac{1}{M} \left(\hat{n}_i^k(t) + \epsilon_c^k\right)\right)^{-\frac{1}{2}} \\
		&\qquad \qquad +\frac{\sigma^2 \Theta_g}{2} \left(\frac{1}{M} \left(\hat{n}_i^k(t) + \epsilon_c^k\right)\right)^{\frac{1}{2}}\\
		&\qquad \quad   \; \&\; (1+\eta)^{g-1} \le \hat n_i^k(t) + \epsilon_c^k< (1+\eta)^g   \Bigg) \\		
		& \ge \mathbb{P} \left( \frac{\hat{s}_i^k(t) - m_i \hat{n}_i^k(t)}{\left(\frac{1}{M} \left(\hat{n}_i^k(t) + \epsilon_c^k\right)\right)^{\frac{1}{2}}} > \sigma \sqrt{\frac{a}{2}} \left( (1+\eta)^{\frac{1}{4}} + (1+\eta)^{-\frac{1}{4}}   \right)  \right),
	\end{align*}
where the last inequality follows from inequality~\eqref{eq:ineq2}. 	}
	Setting $\sigma \sqrt{\frac{a}{2}} \left( (1+\eta)^{\frac{1}{4}} + (1+\eta)^{-\frac{1}{4}} \right) = \delta$, this yields 
	\begin{align*}
		&\mathbb{P} \left( \frac{\hat{s}_i^k(t) - m_i \hat{n}_i^k(t)}{\left(\frac{1}{M} \left(\hat{n}_i^k(t) + \epsilon_c^k\right)\right)^{\frac{1}{2}}} > \delta \right) \\
		& \qquad \leq D \expp{\frac{-2 \delta^2}{\sigma^2 \left( (1+\eta)^{\frac{1}{4}} + (1+\eta)^{-\frac{1}{4}} \right)^2 }}.
	\end{align*}
	It can be verified that the first three terms in the Taylor series for $\frac{4}{\left( (1+\eta)^{\frac{1}{4}} + (1+\eta)^{-\frac{1}{4}} \right)^2}$ provide a lower bound, i.e., 
	\begin{equation*}
		\frac{4}{\left( (1+\eta)^{\frac{1}{4}} + (1+\eta)^{-\frac{1}{4}} \right)^2} \geq 1-\frac{\eta^2}{16}.
	\end{equation*}
	Therefore,	it holds that 
	\begin{align*}
		&\mathbb{P} \left( \frac{\hat{s}_i^k(t) - m_i \hat{n}_i^k(t)}{\left(\frac{1}{M} \left(\hat{n}_i^k(t) + \epsilon_c^k\right)\right)^{\frac{1}{2}}} > \delta \right) \leq D \expp{\frac{- \delta^2}{2\sigma^2}\left( 1-\frac{\eta^2}{16} \right) } \\
		& \qquad \qquad  = \Bigg\lceil \frac{\lnn{t + \epsilon_n}}{\lnn{1+\eta}} \Bigg\rceil  \expp{\frac{- \delta^2}{2\sigma^2}\left( 1-\frac{\eta^2}{16} \right) }.
	\end{align*}
\end{proof}
}
%End large changed section

\section{Cooperative Decision-Making}
\label{DistributedDecisionMaking}
In this section, we extend the UCB algorithm~\cite{PA-NCB-PF:02} to the distributed cooperative setting in which multiple agents can communicate with each other according to a given graph topology. Intuitively, compared to the single agent setting,  in the cooperative setting each agent will be able to perform better due to communication with neighbors.  However,  the extent of an agent's performance advantage depends on the network structure. We compute bounds on the performance of the group in terms of the expected group cumulative regret. We also propose a metric that orders the contribution of agents to the cumulative group regret in terms of their location in the graph.

\subsection{Cooperative UCB Algorithm}
%The UCB algorithm from \cite{auer2002finite} can be extended to the cooperative case  with Gaussian rewards in the form of a new algorithm called Cooperative UCB. 
The  cooperative UCB algorithm is analogous to the UCB algorithm, and uses a modified decision-making heuristic that captures the effect of the additional information an agent receives through communication with other agents as well as the rate of information propagation through the network. 

The cooperative UCB algorithm is initialized by each agent sampling each arm once and proceeds as follows. At each time $t$ each agent $k$ selects the arm with maximum $Q_i^{k}(t-1) = \hat{\mu}_i^{k}(t-1) + C_i^k(t-1)$, { where
%\begin{equation}
%
%\end{equation},
\begin{equation}
C_i^k(t-1) = \sigma \;\sqrt[]{\frac{2 \gamma}{G(\eta)} \cdot \frac{\hat{n}_i^{k}(t-1) +  \epsilon_c^k}{M\hat{n}_i^{k}(t-1)}\cdot\frac{ \lnn{t-1}}{\hat{n}_i^{k}(t-1)}}, \label{Cdefn}
\end{equation}
and receives realized reward $r_i^k(t)$, where $\gamma > 1$, $G(\eta) = (1-{\eta^2}{16})$, and $\eta\in(0,4)$.}  Each agent $k$ updates its cooperative estimate of the mean reward at each arm using the distributed cooperative estimation algorithm described in~\eqref{eqnmean}, \eqref{nhatdefn}, and~\eqref{shatdefn}. Note that the heuristic $Q_i^k$ requires the agent $k$ to know $\epsilon_c^k$, which depends on the global graph structure. This requirement can be relaxed by replacing $\epsilon_c^k$ with an increasing sub-logarithmic function of time. We leave rigorous analysis of the alternative policy for future investigation.

%and 
%
%\begin{align}
%\epsilon_c^k  & =   \frac{M}{\sigma^2} \sum_{p=1}^M \sum_{j=2}^M (\frac{1}{1-|(\lambda_p \lambda_j)|} -1) a_{pj}(k).
%\end{align}
%Additionally, $\lambda_i$ denotes the $i$'th largest eigenvalue of $P$, $\mathbf{u}_i$ denotes the eigenvector corresponding to $\lambda_i$, and $u_i^d$ denotes the $d\in{1,...,M}$ entry of $\mathbf{u}_i$.  $a_{pj}(k)$ as defined in \eqref{apjdefn}. 
\subsection{Regret Analysis of the Cooperative UCB Algorithm}\label{RegretAnalysis}

We now derive a bound on the expected cumulative group regret using the distributed cooperative UCB algorithm.  This bound recovers the upper bound given in \eqref{eqn:fusioncenterregret} within a constant factor.  The contribution of each agent to the group regret is a function of its location in the network.

\begin{theorem}[\bit{Regret of the Cooperative UCB Algorithm}]\label{thm:regret-coop-ucb}
For the cooperative UCB algorithm and the Gaussian multiarmed bandit problem the number of times a suboptimal arm $i$ is selected by all agents until time $T$ satisfies
{
\begin{align*}
\sum_{k=1}^M& \mathbb{E}[n_i^{k}(T)]  \leq  {\max \left\{ M, \bigg \lceil \!  M\epsilon_n \! \! + \!  \sum_{k=1}^M \! \frac{8 \sigma^2 \gamma (1+  \epsilon_c^k) \ln(T)}{M \Delta_i^2} \bigg \rceil \right \} } \\
&\quad +\frac{2M}{\lnn{1+\eta}} \left( \frac{1}{(\gamma-1)^2} + \frac{\lnn{(1+\epsilon_n)(1+\eta)}}{\gamma - 1}  + 2 \right)  \! 
\end{align*}
where $\eta>0$ and $\gamma > 1$.}

%\item the expected cumulative regret of all agents until time $T$ satisfies
%\begin{align}
%\sum_{k=1}^M \supscr{R}{CUCB}_k(t)
%%\sum_{t=0}^T \supscr{R}{CUCB}_k(t) &\nonumber \\ 
%& \leq  \sum_{k=1}^M \sum_{i=1}^N \bigg(\Delta_i \Big( \bigg \lceil M\Big( \epsilon_n & + \frac{\gamma}{\gamma-1}\Big) \nonumber \\
%& \qquad + \Big( \frac{8 \sigma^2 \gamma (1+  \epsilon_c^k) }{M\Delta_i^2}+1 \Big) \ln(T) \bigg \rceil \Big)
%\bigg),\nonumber
%\end{align}
%where $\Delta_i = m_{i^*}-m_i$. 

%where $\beta_i^k(t) =  2\frac{\hat{n}_{i}^{k}(t) + \epsilon_c^k}{\hat{n}_{i}^{k}(t)}.$
\end{theorem}
\begin{proof}
 We proceed similarly to~\cite{PA-NCB-PF:02}.  The number of times a suboptimal arm $i$ is selected by all agents until time $T$ is 
% By upper bounding the probability of this condition occurring, we can compute a upper bound on the expected cumulative regret. Therefore,
\begin{align}
\sum_{k=1}^M & n_i^k(T) = \sum_{k=1}^M\sum_{t=1}^T \mathds{1}(i^k(t) = i^k)\nonumber \\
& \leq  \sum_{k=1}^M \sum_{t=1}^T  \mathds{1}(Q_i^k(t-1) \geq Q_{i^*}^k(t-1))\nonumber \\
& \leq  A + \sum_{k=1}^M \sum_{t=1}^T \mathds{1}(Q_i^k(t-1) \geq Q_{i^*}^k(t-1), \nonumber \\
& \qquad \qquad \qquad \qquad \qquad \quad M n_i^{\text{cent}}(t-1) \geq A),  \label{suboptimal-samples}
\end{align}
where $A >0$ is a constant that will be chosen later.  

At a given time $t+1$ an individual agent $k$ will choose a suboptimal arm only if 
$ Q_i^k(t) \geq Q_{i^*}^k(t)$. 
For this condition to be true at least one of the following three conditions must hold:
\begin{align}
\hat{\mu}^k_{i^*}(t) &\leq m_{i^*} - C_{i^*}^k(t) \label{1stcond} \\
\hat{\mu}^k_{i}(t) &\geq m_{i} + C_{i}^k(t) \label{2ndcond} \\
m_{i^*} &< m_{i} + 2 C_{i}^k(t). \label{3rdcond} 
\end{align}

%where 
%\begin{equation}
%C_i^k(t) = \sqrt[]{\frac{\beta_i^k(t)}{M}\frac{ \log t}{\hat{n}_i^{k}(t)}} \label{Cdefnwbeta}.
%\end{equation}

{
We now bound the probability that \eqref{2ndcond} holds. Applying Theorem \ref{Thm:EstDevBoudnsCondensed},  it follows 
{ for $t \ge N$ (i.e., after initialization) that}
	\begin{align*}
	\mathbb{P}\! \left( \hat{\mu}_{i}^k(t) \!\geq m_{i} \!+\! C_{i}^k(t) \right) & \!\!=\! \mathbb{P} \!\left(\! \frac{\hat{s}_{i}^k - m_{i} \hat{n}_{i}^k}{\sqrt{ \frac{1}{M} \left( \hat{n}_i^{k}(t) +  \epsilon_c^k\right) }} \! \geq \! \sigma \sqrt{\frac{2 \gamma \lnn{t}}{G(\eta)}} \right)\\
	& \leq \Bigg\lceil \frac{\lnn{t + \epsilon_n}}{\lnn{1+\eta}} \Bigg\rceil \expp{-\gamma \lnn{t}} \\
	& \leq \left(  \frac{\lnn{t + \epsilon_n}}{\lnn{1+\eta}} +1\right) \expp{-\gamma \lnn{t}} \\
	& = \left(  \frac{\lnn{t \frac{t+\epsilon_n}{t}}}{\lnn{1+\eta}} +1\right) \frac{1}{t^\gamma} \\
	& \leq \left(  \frac{\lnn{t (1+\epsilon_n)}}{\lnn{1+\eta}} +1\right) \frac{1}{t^\gamma} \\
	& =\left(  \frac{\lnn{t}}{\lnn{1+\eta}} + \frac{\lnn{ 1+\epsilon_n}}{\lnn{1+\eta}} +1\right) \frac{1}{t^\gamma}.
	\end{align*}
{
%Therefore, 
%\[
%\mathbb{P}(\eqref{2ndcond} \text{ holds})   \leq \left(  \frac{\lnn{t-1}}{\lnn{1+\eta}} + \frac{\lnn{ 1+\epsilon_n}}{\lnn{1+\eta}} +1\right) \frac{1}{(t-1)^\gamma}.
%\]

It follows analogously with a slight modification to Theorem \ref{Thm:EstDevBoudnsCondensed} that 
$$
\mathbb{P}(\eqref{1stcond} \text{ holds})   \leq \left(  \frac{\lnn{t}}{\lnn{1+\eta}} + \frac{\lnn{ 1+\epsilon_n}}{\lnn{1+\eta}} +1\right) \frac{1}{t^\gamma}.
$$
}
}	

Finally, we examine the probability that \eqref{3rdcond} holds.  
It follows that
\begin{align*}
m_{i^*} &< m_{i} + 2 C_{i}^k(t)\\
\implies n_i^{\text{cent}}(t) & < \left \lceil \epsilon_n + \frac{8 \sigma^2\gamma  (\hat{n}_i^{k}(t) +  \epsilon_c^k) \ln(t)}{M\Delta_i^2 (\hat{n}_i^{k}(t))^2} \right\rceil \\
& \le \left \lceil \epsilon_n + \frac{8 \sigma^2 \gamma (1+  \epsilon_c^k) \ln(t)}{M\Delta_i^2} \right\rceil .
\end{align*}
From monotonicity of $\ln(t)$, it follows that~\eqref{3rdcond} does not hold if $ n_i^{\text{cent}}(t) \ge  \big\lceil \epsilon_n + \frac{8 \sigma^2 \gamma (1+  \epsilon_c^k) \ln(T)}{M\Delta_i^2} \big\rceil$. 
%$\frac{\hat{n}_i^{k}(t) +  \epsilon_c^k}{\hat{n}_i^{k}(t)}$
% We can thereby manipulate \eqref{3rdcond} to get
%\begin{align}
%\mathbb{P}(&m_{i^*} \leq m_{i} + 2 C_{i}^k(t)) \nonumber \\ 
%& = \mathbb{P}(m_{i^*} - m_i - 2C_i^k(t)< 0) \nonumber \\
%& \leq \mathbb{P}(m_{i^*} - m_i -2 \sigma\sqrt{\frac{(\hat{n}_i^{k}(t) +  \epsilon_c^k) \ln(t)}{M(n_i^{\text{cent}}(t) - \epsilon_n) \hat{n}_i^{k}(t)}} ). \label{3rdcondmaxed}
%\end{align}
%
%\indent The probability of this equation occurring is zero if we choose $n_i^{\text{cent}}(t)$ large enough (i.e. assume we sample suboptimally a certain number of times).  Specifically, the probability of \ref{3rdcondmaxed} is zero if 
%\begin{align*}
%n_i^{k}(T) & \geq \lceil \epsilon_n + \frac{2\sigma^2 \beta_i^k(t) \ln(T)}{M\Delta_i^2} \rceil \\
%& \leq \epsilon_n + \frac{2\sigma^2 \beta_i^k(t) \ln(T)}{M\Delta_i^2} + 1
%\end{align*}
%where $\Delta_i = m_{i^*} - m_i $.
{
Now, let {$A = \max \left\{ M, \lceil M \epsilon_n + \sum_{k=1}^M \frac{8 \sigma^2 \gamma (1+  \epsilon_c^k) \ln(T)}{M \Delta_i^2} \rceil \right\}$, where the first element of the set corresponds to the selection of each arm $i$ once by each player during the initialization and the second element ensures that~\eqref{3rdcond} does not hold.}
Then, it follows from~\eqref{suboptimal-samples} that
\begin{align*}
&\sum_{k=1}^M \mathbb{E}[n_i^{k}(T)]  \! \leq \! {\max \left\{ M, \bigg \lceil \!  M\epsilon_n \! \! + \!  \sum_{k=1}^M \! \frac{8 \sigma^2 \gamma (1+  \epsilon_c^k) \ln(T)}{M \Delta_i^2} \bigg \rceil \right \} }\! \!\\
& \quad  +  \! 2\sum_{k=1}^M \! {\sum_{t=N}^{T-1}} \left(  \frac{\lnn{t}}{\lnn{1+\eta}} + \frac{\lnn{ 1+\epsilon_n}}{\lnn{1+\eta}} +1\right) \frac{1}{t^\gamma} \\
& =  {\max \left\{ M, \bigg \lceil \!  M\epsilon_n \! \! + \!  \sum_{k=1}^M \! \frac{8 \sigma^2 \gamma (1+  \epsilon_c^k) \ln(T)}{M \Delta_i^2} \bigg \rceil \right \} }\\
&\quad + \frac{2M}{\lnn{1+\eta}} \left( \sum_{t=1}^T \frac{\lnn{(1+\epsilon_n)(1+\eta)}}{ t^\gamma} \! + \sum_{t=1}^T \frac{\lnn{t}}{t^\gamma}\!  \right) \\
& \leq  \bigg \lceil  M  \epsilon_n +  \sum_{k=1}^M \frac{8 \sigma^2 \gamma (1+  \epsilon_c^k) }{M \Delta_i^2} \ln(T) \bigg \rceil \\
&\quad + \frac{2M}{\lnn{1+\eta}} \left( \frac{1}{(\gamma-1)^2} +{ \frac{\gamma \lnn{(1+\epsilon_n)(1+\eta)}}{\gamma - 1}  + 1} \right)  .
\end{align*}
This establishes the proof. }

%.
%\end{equation}
%\indent In the spirit of \cite{reverdy2014modeling} this results in 
%\begin{equation}
%\mathbb{E}[n_i^{k}(T)] 
%\end{equation}
%and the cumulative expected regret follows immediately as
%\begin{align}
%\sum_{t=0}^T &R^{C-UCB, k}(t)\nonumber \\ 
%& \leq  \sum_{i=1}^N \Big(\Delta_i ( \epsilon_n + 2 + (1+\frac{2\sigma^2 \beta_i^k(t) }{M\Delta_i^2} )\ln(T))\Big) \nonumber
%\end{align}.
\end{proof}

\begin{remark}[\bit{Towards Explore-Exploit Centrality}] \label{remark:exploreexploitcent}
Theorem~\ref{thm:regret-coop-ucb} provides bounds on the performance of the group as a function of the graph structure.  However, the bound is dependent on the values of $\epsilon_c^k$ for each individual agent.  In this sense, $\varsigma^k \equiv 1/\epsilon_c^k$ can be thought of as a measure of node certainty in the context of explore-exploit problems. For $\epsilon_c^k =0$, the agent behaves like a centralized agent.  Higher values of $\epsilon_c^k$ reflect behavior of an agent with  sparser connectivity. Rigorous connections between $\epsilon_c^k$ and standard notions of network centralities~\cite{MN:10} is an interesting open problem that we leave for future investigation. \oprocend
\end{remark}

\section{Numerical Illustrations}
\label{NetworkPerformanceAnalysis}
In this section, we elucidate  our theoretical analyses from the previous sections with  numerical examples. We first demonstrate that the ordering on the performance of nodes obtained through numerical simulations is identical to the ordering  by our upper bounds. We then investigate the effect of connectivity on the performance of agents in random graphs. 

For all simulations below we consider a $10$-armed bandit problem with mean rewards as in { Table~\ref{RewardsTable}, $\sigma = 30$,  $\eta=0 $, and $\gamma =1$.}  \\
\begin{table}[ht]
\centering
 \caption{Rewards at each arm $i$ for simulations}
\begin{tabular}{ |c||c|c|c|c|c|c|c|c|c|c| }
\hline
$i$ &1&2&3&4&5&6&7&8&9&10\\
\hline
$m_i$ &40 &   50  &  50  &  60 &   70  &  70 &   80 &   90 &   92  &  95\\
 \hline
 \end{tabular}
 \label{RewardsTable}
\end{table}

%Here we analyze the Cooperative UCB algorithm through simulations using the graph structure in Fig. \ref{5AgentGraph} as well as Erdos Renyi random graphs.  Erdos Renyi random graphs are a widely used class of random graphs where any two nodes (representing agents) are connected by a given probability, denoted $\rho$ \cite{bollobas1998random}.  When running simulations using Erdos Renyi random graphs we regenerate the graph for each simulation iteration, and also only use the graph if it is connected.  
% 
%\indent Referring back to \eqref{Cdefn}, one can see that higher values of $\epsilon_c^k$ will result in an increase in exploration by agent $k$.  In general, $\epsilon_c^k$ will be higher for less connected agents, resulting in more exploration and larger cumulative regret.  

\begin{example}[\bit{Regret on Fixed Graphs}]
Consider the set of agents communicating according to the graph in Fig.~\ref{5AgentGraphSim} and using the cooperative UCB algorithm to handle the explore-exploit tradeoff in the distributed cooperative MAB problem. 
%Here we explore the performance of a network of agents connected as in Fig. \ref{5AgentGraph}.  
The values of $\epsilon_c^k$ for nodes $1,2,3,$ and $4$ are $2.31, 2.31, 0,$ and $5.43$, respectively. As predicted by Remark~\ref{remark:exploreexploitcent}, agent $3$ should have the lowest regret, agents $1$ and $2$ should have equal and intermediate regret, and agent $4$ should have the highest regret. These predictions are validated in our simulations shown in Fig.~\ref{5AgentGraphSim}. The expected cumulative regret in our simulations is computed using $500$ Monte-Carlo runs. 
%are listed on Fig. \ref{5AgentGraph}, and one can see that agents with greater degree and more central positioning have lower $\epsilon_c^k$ as expected.  Additionally, lower $\epsilon_c^k$ translates to better performance, as seen in Fig. \ref{5AgentGraphSim}, but does so in a nonlinear fashion, as discussed in Subsection \ref{nonlinearitydiscussion}. 
\end{example}

%\subsection{An Illustrative Example: Analysis of Fig. \ref{5AgentGraph}}
%\indent 

\begin{figure}
	\centering
	\begin{minipage}{.35\textwidth}
		\centering
		\includegraphics[width=.9\linewidth]{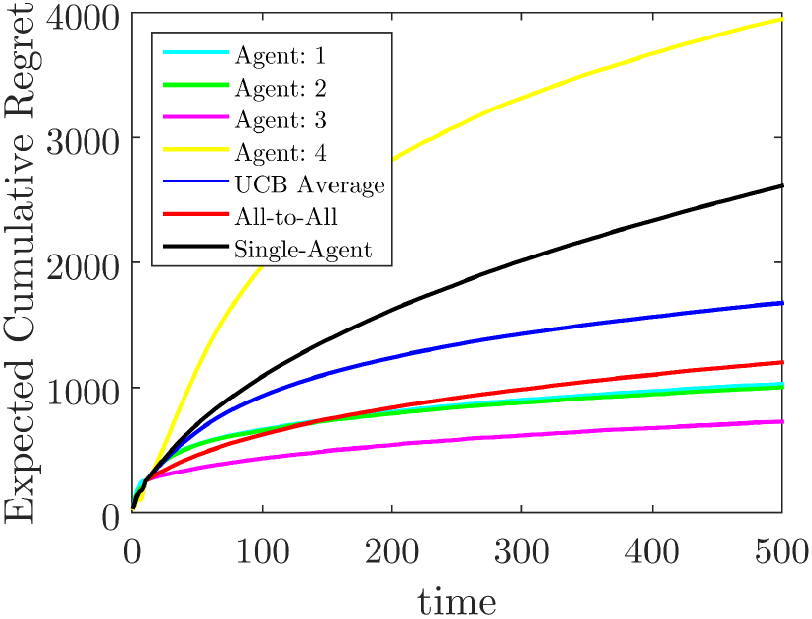}

	\end{minipage}%
	\begin{minipage}{.13\textwidth}
		\centering
		\includegraphics[width=.9\linewidth]{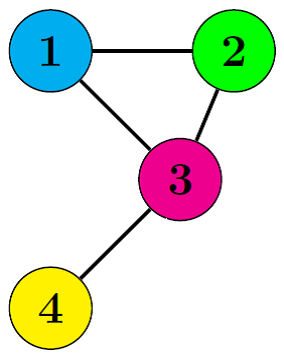}
%		\begin{tikzpicture}[auto, node distance = 1.4cm]
%		
%		\node[SimpleNetwork, fill = cyan, align = center] (1)      {$\mathbf{1}$};
%		\node[SimpleNetwork, fill = green, align = center] (2) [right of=1]  {$\mathbf{2}$};
%		\node[SimpleNetwork, fill = magenta, align = center] (3) [below right of=1]  {$\mathbf{3}$};
%		\node[SimpleNetwork, fill = yellow, align = center] (4) [below left of=3]  {$\mathbf{4}$};
%		
%		\Edge(1)(2)
%		\Edge(1)(3)
%		\Edge(2)(3)
%		\Edge(3)(4)
%		
%		\end{tikzpicture}

	\end{minipage}
\caption{Simulation results comparing expected cumulative regret for different agents in the fixed network shown, using $P$ as in \eqref{Pdefn} and $\kappa = \frac{d_{\text{max}}}{d_{\text{max}}-1}$. Note that agents $1$ and $2$ have nearly identical expected regret.}
\label{5AgentGraphSim}
\end{figure}

%\includegraphics[width=0.35\linewidth]{4agent_graphsim}

%\subsection{Analysis of Erdos Renyi Random Graphs}

We now explore the effect of $\epsilon_c^k$ on the performance of an agent in an Erd{\"o}s-R{\'e}yni random (ER) graph.  ER graphs are a widely used class of random graphs where any two agents are connected with a given probability $\rho$~\cite{bollobas1998random}.  

\begin{example}[\bit{Regret on Random Graphs}]
Consider a set of $10$ agents communicating according to an ER graph and using the cooperative UCB algorithm to handle the explore-exploit tradeoff in the aforementioned MAB problem.  
In our simulations, we consider $100$ connected ER graphs, and for each ER graph we compute the expected cumulative regret of agents using $30$ Monte-Carlo simulations. We show the behavior of the expected cumulative regret of each agent as a function of $\varsigma^k$ in  Fig.~\ref{ER_degcentsim}.
%In our simulations, we do 30 Monte-Carlo runs while keeping the graph fixed and then compute expected cumulative regret for each agent. We repeat this process 100 times, regenerating the graph for each 30 run set. Also, we use the regenerated graph only if it is connected.  

%As shown in Fig. \ref{ER_degcentsim}, 
It is evident that
increased $\varsigma^k$ results in a sharp increase in performance.  Conversely, low $\varsigma^k$ is indicative of very poor performance.  This strong disparity is due to agents with lower $\varsigma^k$ doing a disproportionately high amount of exploration over the network, allowing other agents to exploit.  The disparity is also seen in Fig. \ref{5AgentGraphSim}, as agent $4$ does considerably worse than the others.  Additionally, as shown in Fig.~\ref{5AgentGraphSim} the expected cumulative regret averaged over all agents is higher than the centralized (all-to-all) case.
%We compare the cumulative regret of each agent at a set point in time during the simulation, and then compare this to the degree centrality.  As one can see from Fig. \ref{ER_degcent}, performance improves markedly as degree centrality increases from $1$, with diminishing returns as more connections are made.  
\end{example}

\section{Final Remarks}
\label{FinalRemarks}
Here we used the distributed multi-agent MAB problem to explore cooperative decision-making in networks.  We designed the cooperative UCB algorithm that achieves logarithmic regret for the group.  Additionally, we investigated the performance of individual agents in the network as a function of the graph topology. 
We derived a node certainty measure $\varsigma^k$ that predicts the relative performance of the agents.
% and analyzed performance as a function of a new graph centrality measure $\varsigma^k$.

Several directions of future research are of interest. First, the arm selection heuristic designed in this paper requires  some knowledge of  global parameters. Relaxation of this constraint will be addressed in a future publication. Another interesting direction is to study the tradeoff between the communication and the performance. Specifically, if agents do not communicate at each time but only intermittently, then it is of interest to characterize performance as a function of the communication frequency.

\begin{figure}[ht]
\includegraphics[width=0.35\textwidth]{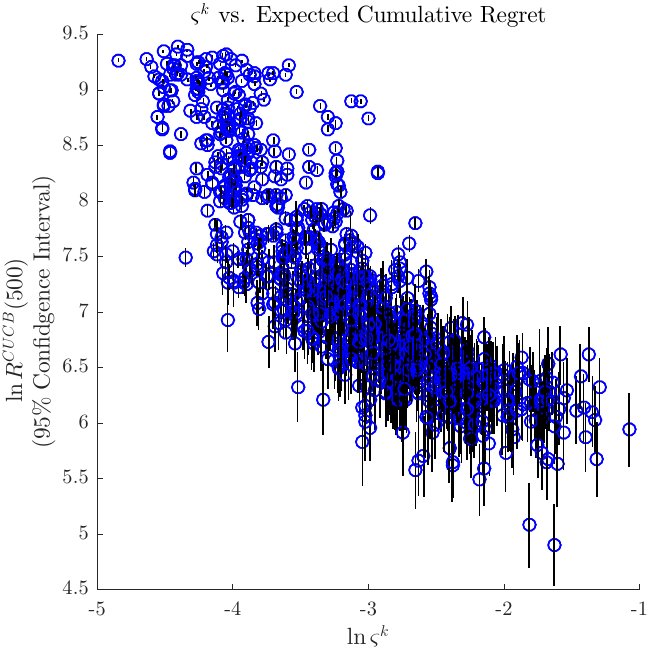}
\centering
\caption{Simulation results of expected cumulative regret as a function of $\varsigma^k$ for nodes in ER graphs with $\rho = \frac{\ln(10)}{10}$,  $P$ as in \eqref{Pdefn}, and $\kappa = \frac{d_{\text{max}}}{d_{\text{max}}-1}$.      }
\label{ER_degcentsim}
\end{figure}

\bibliographystyle{plain}
\bibliography{bibfile,bandits}

\end{document}